\documentclass[conference,a4paper]{IEEEtran}
\usepackage{amsthm,amssymb,amsmath}
\usepackage[cp1251]{inputenc}
\usepackage{graphicx}
\usepackage{epstopdf}
\usepackage{subcaption}
\usepackage{array}
\usepackage{float}
\usepackage[ruled,vlined,linesnumbered]{algorithm2e}
\usepackage{multirow}
\usepackage{multicol}

\theoremstyle{plain}
\newtheorem{lemma}{Lemma}
\newtheorem{theorem}{Theorem}
\newtheorem{definition}{Definition}
\newtheorem{proposition}{Proposition}
\newtheorem{corollary}{Corollary}

\DeclareMathOperator{\wt}{wt}

\DeclareMathOperator{\F}{\mathbb F}
\DeclareMathOperator{\W}{\mathbf W}

\renewcommand{\ell}{l}

\newcommand{\set}[1]{\left\{{#1}\right\}}

\newcommand{\mB}{\mathcal{B}}
\newcommand{\mC}{\mathcal{C}}
\newcommand{\mD}{\mathcal{D}}

\newcommand{\mF}{\mathcal{F}}

\newcommand{\mK}{\mathcal{K}}

\newcommand{\mP}{\mathcal{P}}

\newcommand{\mR}{\mathcal{R}}

\newcommand{\mX}{\mathcal{X}}

\begin{document}

%+Title
%\title{Low complexity decoding of polar codes with some large kernels}
\title{Shortened Polarization Kernels}
\author{
\IEEEauthorblockN{Grigorii Trofimiuk}%, Peter Trifonov}
\IEEEauthorblockA{ITMO University, Russia\\
Email: gtrofimiuk@itmo.ru}}

\maketitle
%-Title

\begin{abstract}
A shortening method for large polarization kernels is presented, which results in shortened kernels with the highest error exponent if applied to kernels of size up to 32. It uses lower and upper bounds on partial distances  for quick elimination of unsuitable shortening patterns. 

The proposed algorithm is applied to some kernels of sizes 16 and 32 to obtain shortened kernels of sizes from 9 to 31. These kernels are used in mixed-kernel polar codes of various lengths. Numerical results demonstrate the advantage of polar codes with shortened large kernels compared with shortened and punctured Arikan polar codes, and polar codes with small kernels.
\end{abstract}

\section{Introduction}

Polar codes are a novel class of error-correcting codes, which achieves the symmetric capacity of a binary-input memoryless channel  $W$. They have low complexity construction, encoding, and decoding algorithms \cite{arikan2009channel}. However, Arikan polar codes allow length $2^t$ only, which  restricts their practical application. To cope with it, many length-compatible constructions were proposed, including shortened and punctured polar codes, chained polar codes 
\cite{trifonov2017chained} and asymmetric construction 
\cite{cavatassi2019asymmetric}.

Polarization is a general phenomenon and is not restricted to Arikan matrix \cite{korada2010polar}.  One can replace it by a larger matrix, called 
\textit{polarization kernel}, which has better polarization properties 
obtaining thus better finite length performance. Polar codes with large 
kernels were shown to provide an asymptotically optimal scaling exponent
\cite{fazeli2018binary}. 
By combining different  kernels, one can construct mixed-kernel (MK) polar code \cite{presman2016mixed}, which may have wide range of possible  lengths. For instance, hybrid design based on kernels of dimension 3 and 5 was suggested in \cite{bioglio2019improved}. Unfortunately, such small kernels have quite poor polarization properties. As a result, these codes do not outperform length-compatible constructions based on Arikan matrix.

At the same time, various  kernels with good properties were recently introduced together with simplified processing algorithms, including $16 \times 16$ and $32 \times 32$ kernels for the window processing algorithm \cite{trofimiuk2021window}, BCH  kernels with reduced trellis processing complexity \cite{moskovskaya2020design}, approximate processing based on Box and Match algorithm \cite{miloslavskaya2014sequentialBCH}, convolutional polar kernels  \cite{morozov2020convolutional} and some kernels \cite{trofimiuk2021search} which allow low complexity processing by recursive trellis processing algorithm \cite{trifonov2021recursive}. However, obtaining polar codes of arbitrary length  requires application of code shortening and puncturing, which are not yet well understood for the case of non-Arikan kernels.

To address this issue we consider shortening of large polarization kernels \cite{korada2010polar}. This allows to obtain kernels of arbitrary smaller size, which can be processed by the same algorithm as the original kernel.  In this paper we propose an algorithm for kernel shortening, resulting in shortened kernels with the highest error exponent  if applied to kernels of size up to 32. The proposed algorithm examines all shortening patterns. Nevertheless, it avoids explicit evaluation of partial distance   for all considered shortening pattern by using lower and upper bounds, substantially reducing thus optimization complexity.

We applied the proposed method to kernels of size $16$ and $32$, obtained shortened kernels of sizes from 9 to 31, and used them in MK polar codes of lengths $32 \cdot l$, $17\leq l \leq 31$. Simulation results show that the obtained codes provide significant performance gain compared with shortened and punctured Arikan polar codes, and polar codes with small kernels. 

\section{Background}
\label{sBackground}
\subsection{Notations}
For a positive integer $n$, we denote by $[n]$ the set $\{0,1,\dots ,n-1\}$. For vectors $a$ and $b$ we denote their concatenation by $a.b$. The vector $u_a^b$ is a subvector $(u_a,u_{a+1},\dots,u_b)$ of a vector $u$. $M[i]$ is the $i$-th row of the matrix $M$. By $M[i,j]$ we denote $j$-th element of $M[i]$. For matrices $A$ and $B$ we denote their Kronecker product by $A \otimes B$. By $A^{\otimes m}$ we denote $m$-fold Kronecker product of matrix $A$\  with itself.

\subsection{Polarizing transformation}
\label{ss:pt}
Consider a binary-input memoryless channel  with transition probabilities $W\{y|c\}, c\in \F_2, y\in \mathcal Y$, where $ \mathcal Y$ is output alphabet. %For a positive integer $n$, denote by $[n]$ the set of $n$ integers $\{0,1,\dots\,n-1\}$. A \textit{polarization kernel} $K$ is a binary invertible $l \times l$ matrix, which is not upper-triangular under any column permutation.
 An $(n = l^m, k)$ polar code is a linear block code generated by $k$ rows of matrix $G_m = M^{(m)}K^{\otimes m}$, where $M^{(m)}$ is a digit-reversal permutation matrix, corresponding to mapping  $ \sum_{i = 0}^{m-1}t_il^i \rightarrow  \sum_{i = 0}^{m-1}t_{m-1-i}l^i$,$t_i \in [l]$.
The encoding  scheme is given by
$ c_0^{n-1}= u_0^{n-1}G_m$,
where $u_i,i\in \mathcal F$ are  set to some pre-defined values, e.g. zero (frozen symbols),  $|\mF| = n - k$, and the remaining values $u_i$ are set to the payload data. 

It is possible to show that a binary input memoryless channel $W$ together with matrix $G_m$ gives rise to bit subchannels $W_{m,K}^{(i)}
(y_0^{n-1},u_0^{i-1}|u_i)$ with capacities approaching $0$ or $1$, and fraction of noiseless subchannels approaching $I(W)$ \cite{korada2010polar}. Selecting  $\mF$ as the set of indices of low-capacity subchannels enables almost error-free communication. 

It is convenient to define probabilities 
%\begin{align}
 %\label{mKernelStepW}
 $\W^{(j)}_{m}(u_0^{j}|y_0^{n-1}) =
 %\frac{W_{m,K}^{(i)}(y_0^{n-1},u_0^{i-1}|u_i)}{2W(y_0^{n-1})}= \!
% \quad \quad \quad \quad \quad \quad\nonumber\\
% \sum_{u_{i+1}^{n-1}}W_{m,K}^{(n-1)}(u_0^{n-1}|y_0^{n-1})= \!
 \sum_{u_{i+1}^{n-1}}\prod_{i = 0}^{n-1}W((u_0^{n-1}G_m)_i|y_i)$.
% \end{align}
where  matrix $G_m$ will be clear from the context. 

Due to the recursive structure of $G_m$,
one has
\begin{align}
\label{mKernProb}
\W^{(sl+t)}_{m}(u_0^{sl+t}|y_0^{n-1}) = 
\quad \quad \quad \quad \quad  \quad \quad \quad \quad \quad \quad \nonumber\\ 
\sum_{u_{sl+t+1}^{l(s+1)-1}} \prod_{j = 0}^{l-1} \W_{m-1}^{(s)}
(\theta_K[u_0^{l(s+1)-1},j]|y_{j\frac{n}{l}}^{(j+1)\frac{n}{l}-1}),
\end{align}
where $\theta_K[u_0^{(s+1)l-1},j]_r = (u_{lr}^{l(r+1)-1}K)_j, r \in [s+1]$.
% Computing these probabilities reduces to soft-output decoding of non-systematically encoded codes generated by last $l-t-1$ rows of $K$. This problem was considered in \cite{griesser2002aposteriori}.

 At the receiver side, the successive cancellation (SC) decoding algorithm makes decisions 
 \begin{equation}
 \label{mSCProb}
 \widehat u_i=\begin{cases}\arg\max_{u_i\in \F_2} \W_m^{(i)}(\widehat u_0^{i-1}.u_i|y_0^{n-1}), &i\notin\mF,\\
\text{the frozen value of $u_i$}&i\in \mF.
\end{cases}
\end{equation}

Since the probabilities \eqref{mKernelWExact} are computed recursively, it is convenient to consider the probabilities for one layer of the polarization transform, which are given by
\begin{equation}
\W_1^{(\phi)}(u_0^\phi|y_0^{l-1})=\sum_{u_{\phi+1}^{l-1} \in \F_2^{l-\phi-1}} \W^{(l-1)}_{1}(u_0^{l-1}|y_0^{l-1}).
% =\sum_{u_{j+1}^{l-1}}\prod_{i = 0}^{l-1}W((u_0^{l-1}K)_i|\mathbf y_{i}).
\label{mKernelWExact}
\end{equation}

The task of computing \eqref{mKernelWExact} is referred to as {\em kernel processing}. The value of $\phi$ is referred to as processing \textit{phase}. Computing these probabilities reduces to soft-output decoding of nonsystematically encoded codes generated by the last $l-\phi-1$ rows of $K$.  
This problem was considered in %rdu
\cite{griesser2002aposteriori}.% and \cite{trifonov2021recursive}.

An $(n, k)$ mixed-kernel polar code \cite{presman2016mixed}, \cite{bioglio2020multikernel} is a linear block code generated by $k$ rows of matrix\footnote{We omit digit-reversal permutation matrix for simplicity} $G_m = \mK_1 \otimes \mK_2\otimes ... \otimes \mK_m$, where $\mK_i$ is an $l_i \times l_i$ polarization kernel and $n = \prod_{i=1}^m l_i$.

\subsection{Fundamental parameters of polar codes}
\subsubsection{Error exponent}
Let $W: \{0,1\} \to \mathcal{Y}$ be a symmetric binary-input discrete memoryless channel (B-DMC) with capacity $I(W)$. By definition,
$$
I(W) = \sum_{y \in \mathcal Y} \sum_{x \in \set{0,1}} \frac{1}{2}W(y|x) \log \frac{W(y|x)}{\frac{1}{2}W(y|0)+\frac{1}{2}W(y|1)}.
$$
%Let $I(W) \in [0,1]$ denote the mutual information between the input and %output of $W$ 
%with uniform distribution on the inputs. 

The Bhattacharyya parameter of $W$ is 
$$Z(W) = \sum_{y\in\mathcal{Y}} \sqrt{W(y|0)W(y|1)}.$$

Consider the polarizing transform $K^{\otimes m}$, where $K$ is an $l \times l$ polarization kernel, and bit subchannels $W_{m}^{(i)}(y_0^{n-1},u_0^{i-1}|u_i)$, induced by it. Let $Z_m^{(i)} = Z(W_{m}^{(i)}(y_0^{n-1},u_0^{i-1}|u_i))$ be a Bhattacharyya parameter of $i$-th subchannel, where $i$ is uniformly distributed on the set $[l^m]$. Then, for any B-DMC $W$ with $0<I(W)<1$,
we  say that an $\ell\times\ell$ matrix $K$ has error exponent (also known as a rate of polarization) $E(K)$ if \cite{korada2010polar}
\begin{itemize}
\item[(i)] For any fixed $\beta < E(K)$,
%\begin{align}\label{def:rate1} 
\[
\liminf_{n \to \infty} \Pr[Z_n \leq 2^{-\ell^{n\beta}}] = I(W).
\]
%\end{align}
\item[(ii)] For any fixed $\beta > E(K)$,
%\begin{align} \label{def:rate2}
\[
\liminf_{n\to\infty} \Pr[Z_n \geq 2^{-\ell^{n\beta}}] = 1.
\]
%\end{align}
\end{itemize}

That is, the error exponent shows how fast bit subchannels of $K^{\otimes m}$ approach either almost noiseless or noisy channel with $n = l^m$. Suppose we construct $(n,k)$ polar code $\mathcal C$ with kernel $K$. Let $P_e(n)$ be a block error probability of $\mathcal C$ under transmission over $W$ and decoding by SC algorithm. It was shown \cite{korada2010polar}, that if $k/n <\ I(W)$ and $\beta < E(K)$, then for sufficiently large $n$ the probability  
$
P_e(n) \leq 2^{-n^\beta}.
$

Let $\langle g_1,g_2,\dots,g_k \rangle$ be a linear code, generated by vectors $g_1,g_2,\dots, g_k$. Let $d_H(a,b)$ be the Hamming distance between $a$ and $b$. Let $d_H(b,\mathcal C) = \min_{c\in \mathcal C} d_H(b,c)$ be a minimum distance between vector $b$ and linear block code $\mathcal C$. Let $\wt(b) = d_H(b,\mathbf 0)$. The \textit{partial distances} (PD) $\mD_i, i \in [l]$, of an $l\times l$ kernel $K$ are defined as follows:
\begin{align}
\label{fPDDef}
\mD_i &= d_H(K[i],\langle K[i+1],\dots,K[l-1] \rangle), i \in [l-1],\\
\mD_{l-1} &= d_H(K[l-1],\mathbf 0).
\end{align}

The vector $\mD$ is referred to as a \textit{partial distances profile} (PDP). It was shown in \cite{korada2010polar} that for any B-DMC $W$ and any $l\times l$ polarization kernels $K$ with partial distances $\{\mD_i\}^{l-1}_{i=0}$, the error exponent $ 
%\begin{align}
E(K) = \frac{1}{l}\sum^{l-1}_{i=0}\log_l \mD_i.
%\label{f:Rate}
%\end{align}
$
%The Arikan kernel 
%$F_t = 
%\begin{pmatrix}
%1&0\\
%1&1
%\end{pmatrix}^{\otimes t}
%$ has rate of polarization $E(F_t) = 0.5$. 

It is possible to show that by increasing kernel dimension $l$ to infinity, one can obtain $l\times l$ kernels with error exponent arbitrarily close to 1  \cite{korada2010polar}. Explicit constructions of kernels with high error exponent are provided in \cite{trofimiuk2021search, presman2015binary,lin2015linear}.
\subsubsection{Scaling exponent}
Another crucial property of polarization kernels is the \textit{scaling exponent}.
Let us fix a binary discrete memoryless channel $W$ of capacity $I(W)$ and a desired block error probability $P_e$. Given $W$ and $P_e$,
suppose we wish to communicate at rate $I(W) - \Delta$ using a family of $(n,k)$ polar codes with kernel $K$. The value of $n$
scales as $O(\Delta^{- \mu(K)})$, where the constant $\mu(K)$ is known as the scaling exponent \cite{hassani2014finitelength}. The scaling exponent depends on the channel. Unfortunately, the algorithm for computing it is only known for the case of binary erasure channel \cite{hassani2014finitelength, fazeli2014scaling}. Furthermore, this algorithm is heuristic. 
It is possible to show \cite{pfister2016nearoptimal,fazeli2021binary} that there exists $l\times l $ kernel $K_l$, such that $\lim_{l\rightarrow \infty}\mu(K_l)=2,$
i.e. the corresponding polar codes provides an optimal scaling behaviour. 
Constructions of the kernels with good scaling exponent are provided in 
\cite{trofimiuk2021search,yao2019explicit}.

\section{Shortening of polarization kernels}
Given a polarization kernel, shortening allows one to obtain from it kernels of smaller size. Moreover, the shortened kernels can be processed by the same processing algorithm as the original kernel, which makes them applicable for construction and decoding of polar codes with a wide range of lengths.

In this section we describe a shortening algorithm, introduce notations, and describe how to process the shortened kernel.

\label{sShortening}
\subsection{Shortening procedure}
\begin{proposition}[\cite{fazeli2014scaling}]
\label{pAddRow}
Let $K$ and $\widehat K$ be $l \times l$ polarization kernels, where $\widehat K$ is obtained from $K$ by adding row $a$ to row $b$ with $a > b$, then $\mC_{\widehat K}^{(i)} = \mC_{K}^{(i)}$ for all $i \in [l]$. Thus, $E(\widehat K) = E(K)$ and $\mu(\widehat K) = \mu(K)$.
\end{proposition}

For convenience, the kernel row additions, described in Proposition \ref{pAddRow} is referred to as \textit{equivalent additions}.

\begin{algorithm}
\caption{\texttt{ShortenSingleCoordinate}$(K, j, l)$}
\label{shortening_step}
$\widehat K \gets K$, $a \gets l-1$; \tcp{initialization}
\For{$i$ \textbf{from} $l-1$ \textbf{downto} $0$}{
\If{$K[i,j] = 1$}
{
\tcp{index of the last 1 in column $j$}
$a \gets i$;\\
\textbf{break};\\
}
}
\tcp{$\widehat K[a,j] = 1$; $\widehat K[i,j] = 0, i \in [l] \setminus \set{a}$}
\For{$i$ \textbf{from} $0$ \textbf{to} $a-1$}{\label{xorLoopStart}
\If{$K[i,j] = 1$}{
$\widehat K[i] \gets \widehat K[i] \oplus \widehat K[a]$ 
}
}\label{xorLoopEnd}
Obtain $\tilde K$ by removing $a$th row and $j$th column from $\widehat K$;\\
\Return $\tilde K$;\\
\end{algorithm}

The shortening of $l \times l$ kernel $K$ on a single coordinate $j$ can be performed by calling \textit{ShortenSingleCoordinate}$(K, j, l)$. This procedure was initially proposed in \cite{korada2010polar}. Observe that due to Proposition \ref{pAddRow}, after steps \ref{xorLoopStart}--\ref{xorLoopEnd} of Alg. \ref{shortening_step} the kernel $\widehat K$ still has the same polarization properties as the kernel $K$. 

For given $l\times l$ kernel $K$ we define a sequence of 
$(l, l-\phi,d_K^{(\phi)})$ \textit{kernel codes} $\mC_{K}^{(\phi)} = \langle K[\phi],\dots,K[l-1] \rangle$, $\phi \in [l],$ and $\mC_K^{(l)}$ contains only zero codeword.

\begin{lemma}
\label{shortened_kernel_codes}
Let $K$ be an $l \times l$ polarization kernel. Let $a$ be the index of the last nonzero row in $j$th column. Let $K' = \textit{ShortenSingleCoordinate}(K, j, l)$. We have
\begin{align}
\label{shortSingle1} \mC^{(i)}_{K'} &= s_{\set{j}}(\mC^{(i)}_{K}), &i \in [a],\\
\label{shortSingle2} \mC^{(i)}_{K'} &= s_{\set{j}}(\mC^{(i+1)}_{K}), &a < i <\ l-1,
\end{align} 
where $s_{\mP}(\mC)$ denotes a code obtained by shortening of linear code $\mC$ on coordinates in $\mP$.
\end{lemma}
\begin{proof}
It directly follows from the definition of shortened linear code and Alg. 
\ref{shortening_step}.
\end{proof}

%Lemma \ref{shortened_kernel_codes} gives us the following
\begin{definition}
A kernel $K'$ is called shortened kernel of $l \times l$ kernel $K$ on position $j$ if it satisfies \eqref{shortSingle1} and \eqref{shortSingle2}.
\end{definition} 

Alg. \ref{shortening_step} can be performed several times to obtain an $(l-t) \times (l-t)$ kernel $K^{(t)}$, $t \in [l-1]$, from $K$. Namely, we define a sequence of kernels 
$
K^{(i+1)} = \textit{ShortenSingleCoordinate}(K^{(i)}, j_{i-1}, l-i),
$
where $i \in [l-1]$ and $K^{(0)} = K$. Let $a_{i} \in [l-i]$ be the index of the removed row of $K^{(i)}$ after shortening on coordinate $j_i$.  We also define the following sequences of vectors: 
\begin{align}
\label{fAArray} A^{(i+1)} &= (A^{(i)}_0, A^{(i)}_1, \dots, A^{(i)}_{a_{i}-1},A^{(i)}_{a_{i}+1}, \dots, A^{(i)}_{l-i-1}),\\ 
J^{(i+1)} &= (J^{(i)}_0, J^{(i)}_1, \dots, J^{(i)}_{j_{i}-1},J^{(i)}_{j_{i}+1}, \dots, J^{(i)}_{l-i-1}),
\end{align}
where $A^{(0)} = J^{(0)} = (0,1,\dots,l-1)$. Therefore, we can map index $j_i \in [l-i]$ of removed from $K^{(i)}$ column to index
\begin{equation}
\label{eqColIndices}
p_i = J^{(i)}_{j_i}, p_i \in [l], i \in [l-t].
\end{equation}
of the column, removed from $K$. Similarly, we introduce the vector $r_0^{l-i-1}$ of removed rows indices from $K$ \begin{equation}
\label{fRIndex}
r_i = A_{a_i}^{(i)}, i \in [l-i].
\end{equation}
% The index $r_i = A_{a_i}^{(i)}$, $r_i \in [l]$, of the removed row of $K$ %after shortening. 
\begin{theorem}
\label{thPermutation}
Let $K^{(t)}_\pi$ be an $(l-t) \times (l-t)$ polarization kernel, obtained by $t$ times application of Alg. \ref{shortening_step} to $l \times l$ kernel $K$. Let $p_{\pi} = (p_{\pi(0)}, p_{\pi(1)}, \dots, p_{\pi(t-1)})$ be a sequence of removed columns indices, defined in 
\eqref{eqColIndices}, $\pi(j)$ is a permutation,. Then, the kernel codes $\mC_{K^{(t)}_\pi}^{(i)}$, $i \in [l-t]$, are the same for any $\pi(j)$. 
\begin{proof}
Lemma \ref{shortened_kernel_codes} implies that for any $\pi(j)$ the kernel codes  
$\mC_{K^{(t)}_\pi}^{(i)} = s_{\mP}(\mC_{K}^{(A_i)})$, $i \in [l-t]$,
where the array $A = A^{(t)}$ is given by \eqref{fAArray} and $\mP = \set{p_0, p_1, \dots, p_{t-1}}$. 
\end{proof}
\end{theorem}

% Alg \ref{shortening_step} implies that 
% Let $V$ be a vector of length $l$, such as $V_i = 1$ if there is $K[i,p_{\pi(j)}] %= 1$ for some $j \in [t]$.  Observe that this vector remains the same for %any permutation $\pi(j)$.  Let $\mR$ be a set of the last $t$ nonzero elements %indices of $V$. Since we use equivalent additions in Alg. \ref{shortening_step}, %the set $\set{r_0,r_1,\dots,r_{t-1}}$ of  removed rows indices of $K$ is %given by $\mR$ for any $\pi(j)$. 
% 
% By equivalent additions the kernel $K$ can be transformed to the kernel %$\widehat K$, such as elements $\widehat K[i, p_{\pi(j)}]$ of columns $p_{\pi(j)}, %j \in [t]$, are nonzero only if $i \in \mR$. Let $\widehat K^{(t)}_\pi$ %be %a kernel obtained by shortening of $\widehat K$ on columns $p_{\pi}$. %From %Alg. \ref{shortening_step} it follows that $\widehat K^{(t)}_\pi$ %is %exactly %the same for any permutation $\pi(j)$, since only rows from %$\mR$ %is modified. %Proposition \ref{pAddRow} implies that $\mC_{\widehat %K}^{(i)} %= \mC_{K}^{(i)}$ %for all $i \in [l]$, therefore, 
% $\mC_{K^{(t)}_\pi}^{(i)} = \mC_{\widehat K^{(t)}_\pi}^{(i)}$ for all $i %\in [l-t]$ and any permutation $\pi$.

One can shorten a  kernel $K$ by consecutive application of Alg. \ref{shortening_step}. Theorem \ref{thPermutation} implies that the kernel codes, and, therefore, the polarization properties of the resulting  kernel $K^{(t)}$ do not depend on the order, in which we shorten the columns of $K$. Thus, we can consider the vector $p$  just as a set $\mP = \set{p_0, p_1, \dots, p_{t-1}}$, which is referred to as a \textit{shortened pattern}. 

\begin{definition}
\label{defShort}
Let $K$ be an $l \times l$ polarization kernel. Let $\mP \subset [l], |\mP| = t,  1\leq|\mP| \leq l-2$. Thus, the kernel $s_{\mP}(K)$  is called a shortened kernel of  $K$ on coordinates  $\mP$ if any kernel code of $s_{\mP}(K)$ is some shortened kernel code of $K$ on coordinates $\mP$. More precisely, $\mC^{(i)}_{s_{\mP}(K)}$, $i \in [l-t]$, is given by $s_{\mP}(\mC_{K}^{(A_i)})$. The array $A = A^{(t)}$, defined in \eqref{fAArray}, denotes the mapping of $s_{\mP}(K)$ row indices $i \in [l-t]$ to row indices $i' \in [l]$ of $K$.
\end{definition}

In addition, by $\mR_\mP(K) = \set{r_0, r_1, \dots, r_{t-1}}$, where $r_i$ is defined by \eqref{fRIndex}, we denote a set of removed rows indices from $K$ after shortening on coordinates $\mP$. 

\subsection{Processing of shortened kernels}

Suppose we have an $l \times l$ polarization kernel $K$ together with some  processing algorithm. We assume that it is given as a black box, i.e. it takes $u_0^{\phi}$ and $y_0^{l-1}$ as input and outputs the probability 
$\W_1^{(\phi)}(u_0^\phi|y_0^{l-1})$ of the kernel $K$. In this section we describe how to modify inputs $u_0^{\phi}$ and $y_0^{l-1}$ to obtain the probabilities $\widetilde \W_1^{(\phi)}(u_0^\phi|y_0^{l-t-1})$ of  shortened kernel $s_{\mP}(K)$. 

Let $\widehat K$ be an $l \times l$ polarization kernel obtained from kernel $K$ by application of equivalent additions, i.e. $\widehat K = T K$, where $T$ is an $l \times l$ upper triangular matrix. Let
$
\widehat \W_1^{(\phi)}(u_0^\phi|y_0^{l-1})
$
be a probability $\phi$-th input symbol of kernel $\widehat K$. 
We have 
$
c_0^{l-1} = v_0^{l-1}K = u_0^{l-1} \widehat K,
$
which implies $v_0^{l-1} = u_0^{l-1} T^{-1}$. Since $T^{-1}$ is also upper triangular, it is can be verified that 
\begin{equation}
\label{WTrans}
\widehat \W_1^{(\phi)}(u_0^\phi|y_0^{l-1}) = \W_1^{(\phi)}(u_0^\phi \cdot (T^{-1})^{(\phi)}|y_0^{l-1}),
\end{equation}
where $(T^{-1})^{(\phi)}$ is an $(\phi + 1) \times (\phi + 1)$ submatrix of $T^{-1}$, which contains its first $\phi+1$ rows and columns. 

We consider shortening of $K$ on columns $\mP$. We start from computing such kernel $\widehat K = TK$ as $\widehat K[i,j] = 0$ for $j \in \mP$ and $i \in ([l] \setminus \mR_{\mP}(K))$. Next, for a given phase $\phi \in [l-t]$, input symbols $u_0^{\phi}$ of $s_{\mP}(K)$ and input LLRs $y_0^{l-t-1}$  we define an extended phase $\psi = A^{(t)}_\phi$, an extended input vector $\tilde u_0^{\psi}$ and an extended channel output vector $\tilde y_0^{l-1}$. More specifically, $\tilde u_{A^{(t)}_i} = u_i$ for $i \in [\phi]$, $\tilde u_j = 0$ for $j \in \mR_{\mP}(K)$,
and $\tilde y_{A^{(t)}i_i} = y_i$ for $i \in [\phi]$, $\tilde y_j = w_0, W(0|w_0) = 1$, for $j \in \mR_{\mP}(K)$. That is, the probabilities of the shortened kernel are given by
\begin{equation}
\widetilde \W_1^{(\phi)}(u_0^\phi|y_0^{l-t-1}) = 
\W_1^{(\psi)}(\tilde u_0^\psi \cdot (T^{-1})^{(\psi)}|\tilde y_0^{l-1}).
\end{equation}

% \begin{example}
% \label{ExampleK5}
% Consider $5 \times 5$ kernel 
% \scalebox{0.85}{
% $K = \left(
% \arraycolsep=1.35pt\def\arraystretch{0.7}
% \begin{array}{ccccc}
% 1&0&0&0&0\\
% 1&1&0&0&0\\
% 1&0&1&0&0\\
% 1&1&1&1&0\\
% 1&0&0&0&1\\
% \end{array}
% \right).$
% }
% 
% Shortening of $K$ on 4th column results in kernel
% \scalebox{0.85}{
% $s_{\set{4}}(K) = 
% \left(
% \arraycolsep=1.35pt\def\arraystretch{0.7}
% \begin{array}{cccc}
% 1&0&0&0\\
% 1&1&0&0\\
% 1&0&1&0\\
% 1&1&1&1\\
% \end{array}
% \right)$ 
% }
% with $E(s_{\set{4}}(K)) = 0.5$, wile  shortening on 3th column results in kernel
% \scalebox{0.85}{
% $s_{\set{3}}(K) = 
% \left(
% \arraycolsep=1.35pt\def\arraystretch{0.7}
% \begin{array}{ccccc}
% 1&0&0&0\\
% 1&1&0&0\\
% 1&0&1&0\\
% 1&0&0&1\\
% \end{array}
% \right)$} 
% with $E(s_{\set{3}}(K)) = 0.375$.
% \end{example}

\section{Optimizing the shortening pattern}
Error exponent of the shortened kernel can substantially vary  for different shortening patterns. Our goal is to find such a shortening pattern that results in a shortened kernel with as large error exponent as possible. 

\subsection{Bounds on partial distances after shortening}

%The error exponent $E(s_\mP(K))$ is hard to compute for an arbitrary $K$ %and $\mP$. In this section we provide lower and upper bounds on partial %distances of $s_\mP(K)$, which allows us to quickly estimate $E(s_\mP(K))$.

\begin{lemma}[\cite{korada2010polar}]
\label{lemmaPD}
Let $K$ be an $l \times l$ polarization kernel with partial distances profile (PDP) $\mD$. Let $\mD'$ be a PDP  of a shortened kernel $s_{\set{j}}(K)$, and $\mR_{\set{j}}(K) = \set{a}$. We have
\begin{align}
&\mD'_k \geq \mD_k, &0 \leq k \leq a -1, \\
&\mD'_k = \mD_{k+1}, &a \leq k \leq l-2.
\end{align}
\end{lemma}

It turns out that %the bound of 
Lemma \ref{lemmaPD} can be refined in some cases.%improved ?
\begin{lemma}
\label{lemmaWT}
For any $l \times l$ polarization kernel $K$ with PDP $\mD$ we have
%\begin{equation}
$\mD_i \leq \wt(K[i]), i \in [l]$.
%\end{equation}
\end{lemma}
\begin{proof}
If follows directly from definition \eqref{fPDDef}.
\end{proof}

By $\mX_{\mP}(K)$ we denote a set of $s_{\mP}(K)$ row indices, which were modified by equivalent additions during the shortening. 
In particular, let $a$ be an index of the last nonzero row in column $j$, thus, 
$\mX_{\set{j}}(K) = \set{b|K[b,j] =1, b \in [a]}$.
\begin{lemma}
\label{lemmaPDW}
Consider an $l \times l$ polarization kernel $K$ with PDP $\mD$, such that $\wt(K[i]) = \mD_i, i \in [l]$. Let $a$ be an index of last nonzero row in column $j$ of $K$. Thus, for PDP $\mD'$ of a kernel $K ' =s_{\set{j}}(K)$ we have  
\begin{align}
&\label{L9} \mD'_k = \mD_{k+1}, &(a \leq k \leq l-2), \\ 
&\label{L10} \mD_k \leq \mD'_k \leq \wt(K'[k]), &k \in \mX_{\set{j}}(K)\\
&\label{L11} \mD'_k = \mD_k, &k \in ([a] \setminus \mX_{\set{j}}(K))  \
\end{align}
\end{lemma}
\begin{proof}
Terms \eqref{L9}--\eqref{L10} follows from Lemmas \ref{lemmaPD} and \ref{lemmaWT}. By Lemma \ref{lemmaPD} we have $\mD'_k \geq \mD_k$, $k \in [a]$. For 
$k \in ([a] \setminus \mX_{\set{j}}(K))$ we have $\wt(K'[k]) = \wt(K[k]) = \mD_k$, thus, by Lemma \ref{lemmaWT} we have $\mD_k \geq \mD'_k$, which results in \eqref{L11}.
\end{proof}
Application of Lemma \ref{lemmaPDW} to the case of shortening on $t$ columns results in the following

\begin{corollary}
\label{CorD}
Consider an $l \times l$ polarization kernel $K$ with PDP $\mD$, such that $\wt(K[i]) = \mD_i, i \in [l]$. Thus, for partial distances $\mD'_i, i \in [l-t]$ of  $K' =s_{\mP}(K)$, $|\mP| = t$, we have 
\begin{align}
&\mD'_i = \mD_{A_i}, &A_i \not \in \mX_{\mP}(K),\\
\label{eqIndToCompute}\mD_{A_i} \leq &\mD'_i  \leq \wt(s_{\mP}(K)[i]), &A_i  \in \mX_{\mP}(K),
\end{align}
where $A$ is given in definition \ref{defShort}.
\end{corollary} 

\subsection{Algorithm of searching for optimal shortening pattern}

In this section we propose an algorithm to find a shortening pattern resulting in a kernel with the largest error exponent.

Suppose we want to shorten $l \times l$ kernel $K$ on $t$ coordinates. Thus, using Theorem \ref{thPermutation}, we can formulate our task as finding

\begin{equation}
\label{short_max}
\mP^{\star} = \arg  \max_{ \mP \in \mB(l,t)} E(s_{\mP}(K)),
\end{equation}
where $\mB(l,t) = \set {B|B \subseteq [l], |B| = t}, |\mB(l,t)| = \binom{n}{k}.$
% is a set containing all size-$t$ subsets of $[l]$. 

\begin{algorithm}[ht]
\caption{\texttt{FindOptimalShortening}$(K, l, t)$}
\label{find_optimal_shortening}
Compute PDP $\mD$ of $K$;\\
Transform $K$ into 
$K'$ with $\wt(K'[i]) = \mD_i$ by equivalent additions;\\
$E^{\star} = 0; \mP^{\star} = \emptyset; l' \gets l-t$;\\
\For{\textbf{each} $\mP$ \textbf{in} $\mB(l,t)$}{ \label{algForeachLB}
Compute $s_{\mP}(\widetilde K)$;\\
\tcc{lower bound on error exponent}
$\widehat E \gets \frac{1}{l'} \sum_{i \in ([l] \setminus {\mR_{\mP}(\widetilde K)})}\log_{l'}\mD_i$;\\
\If{$\widehat E > E^{\star}$}{ 
$E^{\star} \gets \widehat E; \mP^{\star} \gets \mP$;\label{algForeachLBEnd}\\
}
  
}
\For{\textbf{each} $\mP$ \textbf{in} $\mB(l,t)$}{ \label{algForeach}
\tcc{Upper bound from Corollary \ref{CorD}}
Compute $\bar \mD_0^{l'-1}$, where 
$
\bar \mD_i = 
\begin{cases}
\wt(s_{\mP}(K')[i]), &A_i \in \mX_{\mP}(K'),\\
\mD_{A_i}, &\text{otherwise}; 
\end{cases}
$\\
$\bar E \gets \frac{1}{l'}\sum_{i = 0}^{l'-1} \log_{l'}\bar \mD_i$;\\
\If{$\bar E < E^{\star}$}{
\textbf{continue};\\
}
Compute PDP\ $\mD'$ of $s_{\mP}(K')$;\\
$E \gets \frac{1}{l'}\sum_{i = 0}^{l'-1} \log_{l'} \mD'_i$;\label{algComputePD}\\
% \For{$i \in \set{j|j \in [l'], A^{(t)}_j \in \mX_{\mP}(\widetilde K)}$}{
% Compute $i$th partial distance $\mD'_i$ of $s_{\mP}(\widetilde K)$;\\ 
% }
\If{$E >\ E^{\star}$}{
$E^{\star} = E, \mP^{\star} = \mP$;\\
}
}
\Return $\mP^{\star}$;\\
\end{algorithm}

To compute \eqref{short_max} one should run \textit{FindOptimalShortening}$(K, l, t)$ described in Alg. \ref{find_optimal_shortening}.
Note that the complexity of $s_{\mP}(K)$ computation is $O(n\cdot|\mP|)$ row additions. Thus, we propose to initially compute lower bound on the optimal error exponent, which is done at lines \ref{algForeachLB}-\ref{algForeachLBEnd} of Alg. \ref{find_optimal_shortening}. This allows us to skip costly calculation of $s_{\mP}(K)$ PDP if the upper bound $\bar E$ is less than current $E^{\star}$. Moreover, at line \ref{algComputePD} one could compute partial distances $\mD'_i$ only for indices $i$ given in \eqref{eqIndToCompute}.

Observe that the proposed Alg. \ref{find_optimal_shortening} is feasible for $l \leq 32$. For larger $l$ one can pick a random subset of $\mB(l,t)$ at line \ref{algForeachLB}, but in this case Alg. \ref{find_optimal_shortening} is not guaranteed to find an optimal shortening pattern \eqref{short_max}. Nevertheless, one can still estimate lower bound on $E^{\star}$ at lines \ref{algForeachLB}-\ref{algForeachLBEnd} of Alg. \ref{find_optimal_shortening} even for large $l$.

% Let $\bP$ $=$ $\set{\mP|\mP \in \mB(l,t), E(s_{\mP}(K)) = E^{\star}}$, where $E^{\star}$ $=$ $E(s_{\mP^{\star}}(K))$. We observed that shortening patterns from $\bP$ may result in kernels with different scaling exponents, thus one can compute the scaling exponent of kernels, shortened on patterns from $\bP$. 

\subsection{Shortening of some known kernels}

In this work we consider MK polar codes of length $l \cdot 32$ for $17 \leq l \leq 31$. Recently, $K_{32}$  kernel with $E(K_{32}) = 0.522$ and $\mu(K_{32}) = 3.417$, together with efficient window processing algorithm was introduced in \cite{trofimiuk2021window}.  So we basically consider $K_{32} \otimes F_{5}$ polarizing transform, where
$
F_t = \left(
\arraycolsep=1.15pt\def\arraystretch{0.5}
\begin{array}{cc}
1&0\\
1&1
\end{array} \right)^{\otimes t}
$ is an $2^t \times 2^t$ Arikan matrix. Therefore, to obtain codes of length $l \cdot 32, 17 \leq l \leq 31,$ we shorten either $K_{32}$ or $F_{5}$, which results in polarizing transforms $s_{\mP}(K_{32}) \otimes F_{5}$ or $K_{32} \otimes s_{\mP}(F_5)$. In addition, we shortened kernels $F_{4}$, $K_{16}$, which was also proposed for window processing algorithm \cite{trofimiuk2021window}, and  $32 \times   32$ BCH kernel $B_{32}$ proposed in  \cite{moskovskaya2020design}.  

% Please add the following required packages to your document preamble:
% \usepackage{multirow}
% Please add the following required packages to your document preamble:
% \usepackage{multirow}
\begin{table}[htbp]
\caption{Properties of the shortened polarization kernels}
\label{tShort32}
\centering
\scalebox{0.9}{
\footnotesize
\setlength{\tabcolsep}{1.25pt}
\begin{tabular}{
% |@{\hspace{0.5mm}}>{\centering}p{1.7mm}@{\hspace{1.5mm}}
% ||@{\hspace{0.3mm}}>{\centering}p{5.3mm}@{\hspace{1.3mm}}
% |@{\hspace{0.3mm}}>{\centering}p{5.3mm}@{\hspace{1.3mm}}
% |@{\hspace{0.1mm}}>{\centering}p{13mm}@{\hspace{1mm}}
% ||@{\hspace{0.3mm}}>{\centering}p{5.3mm}@{\hspace{1.3mm}}
% |@{\hspace{0.3mm}}>{\centering}p{5.3mm}@{\hspace{1.3mm}}
% |@{\hspace{0.1mm}}>{\centering}p{12mm}@{\hspace{1mm}}
% ||@{\hspace{0.3mm}}>{\centering}p{5.3mm}@{\hspace{1.3mm}}
% |@{\hspace{0.3mm}}>{\centering}p{5.3mm}@{\hspace{1.3mm}}
% ||p{6mm}
% |
|c||c|c|c||c|c|c||c|c||c|
}
\hline
\multirow{2}{*}{$l$} & \multicolumn{3}{c||}{$F_5$}                  & \multicolumn{3}{c||}{$K_{32}$}               & \multicolumn{2}{c||}{$B_{32}$} & \multirow{2}{*}{$E$,   \cite{korada2010polar}} \\ \cline{2-9}
                     & $E^{\star}$ & $\mu^{\star}$ & $\mP^{\star}$ & $E^{\star}$ & $\mu^{\star}$ & $\mP^{\star}$ & $E^{\star}$                           & $\mu^{\star}$                          &                                                \\ \hline
17                   & 0.475       & 3.907         & FF00FE00      & 0.492       & 3.693         & FFFE0000      & 0.492                                 & 3.616                                  & 0.492                                          \\ \hline
18                   & 0.466       & 3.993         & F0E0F0E0      & 0.490       & 3.652         & FFF0000C      & 0.490                                 & 3.568                                  & 0.490                                          \\ \hline
19                   & 0.458       & 4.119         & F0E0E0E0      & 0.480       & 3.815         & FFF80000      & \textbf{0.494}                                 & 3.499                                  & 0.487                                          \\ \hline
20                   & 0.463       & 4.036         & C8C8C8C8      & 0.490       & 3.684         & FFF00000      & 0.497                                 & 3.450                                  & 0.497                                          \\ \hline
21                   & 0.455       & 4.168         & F0C0E0C0      & 0.481       & 3.847         & FEF00000      & \textbf{0.493}                                 & 3.467                                  & 0.487                                          \\ \hline
22                   & 0.459       & 4.138         & F0C0C0C0      & 0.488       & 3.789         & A888A888      & 0.494                                 & 3.413                                  & 0.494                                          \\ \hline
23                   & 0.461       & 4.136         & C8888888      & 0.489       & 3.779         & C8888888      & 0.501                                 & 3.415                                  & 0.501                                          \\ \hline
24                   & 0.473       & 3.955         & 88888888      & 0.499       & 3.624         & 88888888      & \textbf{0.51}                                & 3.308                                  & 0.504                                          \\ \hline
25                   & 0.465       & 4.075         & F000E000      & 0.490       & 3.761         & F000E000      & \textbf{0.505}                                 & 3.326                                  & 0.500                                          \\ \hline
26                   & 0.466       & 4.071         & C080C080      & 0.495       & 3.731         & F000A000      & \textbf{0.509}                                 & 3.301                                  & 0.505                                          \\ \hline
27                   & 0.467       & 4.091         & C0808080      & 0.495       & 3.747         & F8000000      & \textbf{0.51}                                 & 3.445                                  & 0.508                                          \\ \hline
28                   & 0.475       & 3.973         & F0000000      & 0.502       & 3.651         & F0000000      & 0.515                                 & 3.279                                  & 0.515                                          \\ \hline
29                   & 0.476       & 4.001         & C0008000      & 0.500       & 3.704         & E0000000      & 0.517                                 & 3.289                                  & 0.517                                          \\ \hline
30                   & 0.482       & 3.912         & C0000000      & 0.506       & 3.647         & C0000000      & 0.522                                 & 3.272                                  & 0.522                                          \\ \hline
31                   & 0.488       & 3.838         & 80000000      & 0.511       & 3.596         & 80000000      & 0.526                                 & 3.264                                  & 0.526                                          \\ \hline
32                   & 0.5         & 3.627         &               & 0.522     & 3.417         &               & 0.537                                 & 3.122                                  & 0.537                                          \\ \hline
                     & \multicolumn{3}{c||}{$F_4$}                  & \multicolumn{3}{c||}{$K_{16}$}               & \multicolumn{3}{l|}{\multirow{9}{*}{}}                                                                                          \\ \cline{1-7}
9                    & 0.456       & 4.129         & F0E0          & 0.462       & 3.960         & F281          & \multicolumn{3}{l|}{}                                                                                                           \\ \cline{1-7}
10                   & 0.452       & 4.185         & C8C8          & 0.462       & 3.876         & FC00          & \multicolumn{3}{l|}{}                                                                                                           \\ \cline{1-7}
11                   & 0.447       & 4.333         & C888          & 0.477       & 3.885         & F800          & \multicolumn{3}{l|}{}                                                                                                           \\ \cline{1-7}
12                   & 0.465       & 4.063         & 8888          & 0.492       & 3.676         & F000          & \multicolumn{3}{l|}{}                                                                                                           \\ \cline{1-7}
13                   & 0.457       & 4.227         & C080          & 0.482       & 3.882         & E000          & \multicolumn{3}{l|}{}                                                                                                           \\ \cline{1-7}
14                   & 0.469       & 4.088         & C000          & 0.491       & 3.810         & C000          & \multicolumn{3}{l|}{}                                                                                                           \\ \cline{1-7}
15                   & 0.478       & 4.009         & 8000          & 0.498       & 3.773         & 8000          & \multicolumn{3}{l|}{}                                                                                                           \\ \cline{1-7}
16                   & 0.5         & 3.627         &               & 0.5         & 3.450         &               & \multicolumn{3}{l|}{}                                                                                                           \\ \hline
\end{tabular}
}
\end{table}

Table \ref{tShort32} presents the optimal shortening patterns and polarization properties for above mentioned kernels. The shortening patterns are provided in big-endian hexadecimal representation of the number $\sum_{p \in \mP} 2^p$. For comparison, we included the error exponents of BCH kernels, shortened by the greedy algorithm \cite{korada2010polar}. It can be observed that in many cases the proposed Alg. \ref{find_optimal_shortening} results in shortened BCH kernels with higher error exponent compared with the greedy algorithm.
 
It turns out that different shortening patterns may result in polarization kernels with the same error exponent but different scaling exponents. Thus, for each kernel and shortening size we tried to find a shortening pattern resulting in the minimal scaling exponent. Note that the values $\mu^{\star}$ from Table \ref{tShort32} are most likely not optimal. Unfortunately, explicit minimization of scaling exponent after shortening is an open problem.

\section{Numeric results}

The performance of all codes was investigated for the case of AWGN\ channel with BPSK\ modulation. The sets of frozen symbols were obtained by method proposed in \cite{trifonov2019construction}. 

\begin{figure}[ht]
\includegraphics[width=\linewidth]{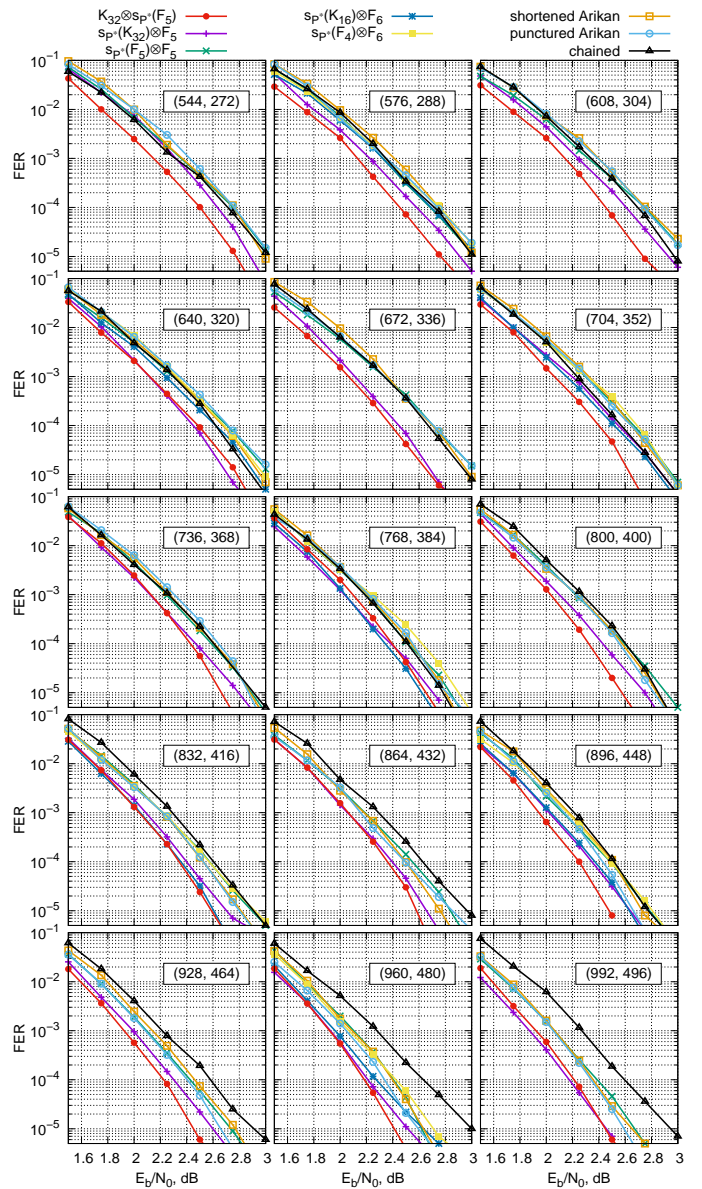}
\caption{Performance of mixed-kernel polar subcodes with shortened kernels}
\label{figShortenedKernels}
\end{figure}

Fig. \ref{figShortenedKernels} presents the simulation results for MK polar subcodes (PSCs) \cite{trifonov2017randomized}, chained \cite{trifonov2018randomized}, shortened and punctured PSCs of different lengths with rate $1/2$. All codes were decoded by successive cancellation list (SCL) decoder \cite{tal2015list} with the list size 8. It can be observed that MK PSCs with $K_{32} \otimes s_{\mP^{\star}}(F_5)$  provide approximately 0.2 dB gain compared with all polar code constructions based on Arikan kernel. %Recall that the corresponding shortened patterns are provided in Table \ref{tShort32}. 

It is noticeable that MK PSCs of lengths $12 \cdot 64 = 24 \cdot 32 = 768$ and 
$13 \cdot 64 = 26 \cdot 32 = 832$ with  
$s_{\mP^{\star}}(K_{16}) \otimes F_6$ has the same performance as the codes with $K_{32} \otimes s_{\mP^{\star}}(F_5)$. This fact shows that kernels of size $l < 16$ and even $E < 0.5$ are useful for MK polar code construction. 

It can be seen that for some lengths MK PSCs with  
 $K_{32} \otimes s_{\mP^{\star}}(F_5)$ significantly outperforms codes with 
$s_{\mP^{\star}}(K_{32}) \otimes F_5$.  To the best of our knowledge, there are no analytical tools for analysis of MK polar codes with arbitrary kernels. The problem of MK polar code construction, which includes the selection of kernels included in the transform and its permutation, is an open problem and out of scope of this paper

Observe that one can obtain polar codes of considered lengths by shortening of the entire polar code with polarizing transform $K_{32} \otimes F_{5}$. Unfortunately, there are no efficient methods for shortening of polar codes with arbitrary large kernels. Moreover, there are no methods for efficient estimation of bit subchannels reliability for such codes. On the contrary, method \cite{trifonov2019construction} allows one to estimate the bit subchannel capacities of MK polar codes with transform $\mK_1^{\otimes t_1} \otimes \mK_2^{\otimes t_2}\otimes ... \otimes \mK_m^{\otimes t_m}$, $t_i \geq 1$,  if subchannel capacity function is provided for each kernel $\mK_i, i \in m$. To compare the performance of these two approaches, we included the simulation results for transforms $s_{\mP}(F_5) \otimes F_5$  and $s_{\mP}(F_4) \otimes F_6$. It can be observed from Fig. \ref{figShortenedKernels}, that the performance of the shortened Arikan polar subcodes and polar subcodes with shortened Arikan kernels are similar.

\begin{figure}[htbp]
\centering
\resizebox{!}{0.311\textwidth}{\input{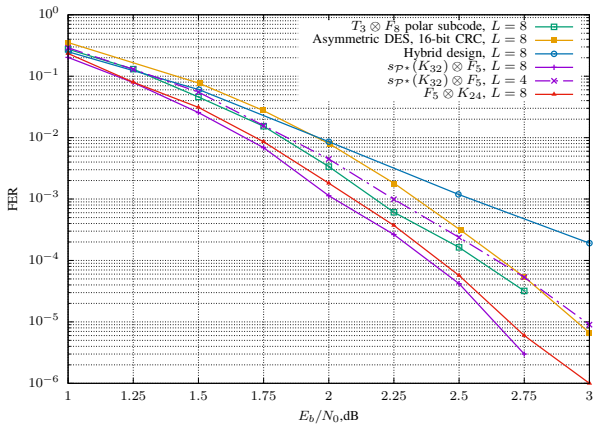}}
\caption{Performance of $(768, 384)$ polar codes}
\label{f768}
\end{figure}

Fig. \ref{f768} illustrates the  performance of different $(768, 384)$ polar codes decoded by SCL decoder with list size $L$. We report the results for $s_{\mP^{\star}}(K_{32}) \otimes F_5$,  $F_5 \otimes K_{24}$ and $T_3 \otimes F_8$ MK PSCs, where $K_{24}$ is a $24 \times 24$ kernel and $T_3$ is a $3\times 3$ kernel introduced in \cite{trofimiuk2021searchArxiv} and \cite{bioglio2020multikernel} respectively . We also include the results for descending (DES) asymmetric construction with 16-bit CRC \cite{cavatassi2019asymmetric} and improved hybrid design \cite{bioglio2019improved} of MK polar codes. It can be seen that MK polar code with the proposed shortened kernel $s_{\mP^{\star}}(K_{32})$ outperforms all other code constructions decoded with the same list size ($L= 8$) and demonstrates the same performance under SCL with $L = 4$. 

The $K_{24}$ kernel was obtained by recently introduced search algorithm \cite{trofimiuk2021search}, which can also be used to obtain polarization kernels which admit low processing complexity. However, for each desired size of the kernel, this method requires complicated search over different partial distances. Moreover such kernels does not have a common structure, whereas the proposed shortening method can be used to obtain a kernel of an arbitrary size (less than base kernel) and shortened polarization kernels can be processed by the same algorithm.

\section{Conclusions}

In this paper a shortening method for large polarization kernels was proposed. This algorithm uses lower and upper bounds on partial distances after shortening to quickly eliminate unsuitable shortening patterns. The algorithm results in a shortened polarization kernel with the highest error exponent if applied to kernels of size up to 32.

We demonstrated  application of the proposed algorithm to some $16 \times 16$ and $32 \times 32$ kernels and obtained kernels of size from 9 to 31. Polar codes based on the obtained kernels were shown to outperform shortened and punctured polar codes with Arikan kernel, and polar codes with small kernels. 

\section*{Acknowledgement} 
This work is partially supported by the Ministry of Science and Higher Education of Russian Federation, passport of goszadanie no. 2019-0898.

\bibliographystyle{ieeetran}
%\bibliography{coding,comm,math,misc,trifonov,miloslavskaya,trofimiuk,morozov}
% Generated by IEEEtran.bst, version: 1.14 (2015/08/26)

\end{document}